\documentclass[10pt,twocolumn]{IEEEtran}

\usepackage{amsmath, mathrsfs}
\usepackage{amsfonts}
\usepackage{amssymb}
\usepackage{graphicx}
\usepackage{color}
\usepackage{multirow}
\usepackage{graphicx}

\usepackage[
	a4paper,
        left=1.6cm,
        right=1.6cm,
        top=1.8cm,
        bottom=1.8cm]
{geometry}

\newcommand{\subparagraph}{}
\usepackage[compact]{titlesec}
\titlespacing*{\section}{0pt}{0.8\baselineskip}{1\baselineskip}

\allowdisplaybreaks

\usepackage{amsfonts}
\usepackage{amssymb}
\usepackage{graphicx}
\usepackage{color}
\usepackage{multirow}
\usepackage{graphicx}

\usepackage{caption}
\usepackage{subcaption}
\usepackage[numbers,sort&compress]{natbib}

\usepackage{multicol}


\def\mindex#1{\index{#1}}

\def\sq{\hbox{\rlap{$\sqcap$}$\sqcup$}}
\def\qed{\ifmmode\sq\else{\unskip\nobreak\hfil
\penalty50\hskip1em\null\nobreak\hfil\sq
\parfillskip=0pt\finalhyphendemerits=0\endgraf}\fi\medskip}

\long\def\defbox#1{\framebox[.9\hsize][c]{\parbox{.85\hsize}{%
\parindent=0pt
\baselineskip=12pt plus .1pt      
\parskip=6pt plus 1.5pt minus 1pt 
 #1}}}

\long\def\beginbox#1\endbox{\subsection*{}%
\hbox{\hspace{.05\hsize}\defbox{\medskip#1\bigskip}}%
\subsection*{}}

\def\endbox{}

\def\diag{{\text{diag}}}

\def\supp{{\rm supp\,}}

\newsavebox{\junk}
\savebox{\junk}[1.6mm]{\hbox{$|\!|\!|$}}

\def\argmin{\mathop{\rm arg\, min}}

\def\Re{\field{R}}

\def\bC{{\mathbb C}}

\def\bE{{\mathbb E}}

\def\bL{{\mathbb L}}

\def\bP{{\mathbb P}}

\def\bR{{\mathbb R}}

\def\bT{{\mathbb T}}

\def\bZ{{\mathbb Z}}

\def\bfG{{\bf G}}
\def\bfH{{\bf H}}
\def\bfI{{\bf I}}

\def\bfK{{\bf K}}

\def\bfX{{\bf X}}

\def\bfa{{\bf a}}
\def\bfb{{\bf b}}

\def\bfs{{\bf s}}

\def\bfu{{\bf u}}

\def\bfw{{\bf w}}
\def\bfx{{\bf x}}
\def\bfy{{\bf y}}
\def\bfz{{\bf z}}

\def\scrC{{\mathscr{C}}}

\def\scrI{{\mathscr{I}}}

\def\scrK{{\mathscr{K}}}

\def\scrR{{\mathscr{R}}}

\def\sfH{{\sf H}}

\def\sfK{{\sf K}}

\def\bfmath#1{{\mathchoice{\mbox{\boldmath$#1$}}%
{\mbox{\boldmath$#1$}}%
{\mbox{\boldmath$\scriptstyle#1$}}%
{\mbox{\boldmath$\scriptscriptstyle#1$}}}}

\def\bfmY{\bfmath{Y}}

\def\bfmhhaY{\bfmath{\hhaY}} 
\def\bfmhhaY{\hbox to 0pt{$\widehat{\bfmY}$\hss}\widehat{\phantom{\raise 1.25pt\hbox{$\bfmY$}}}}

\def\til={{\widetilde =}}

\def\clA{{\cal A}}
\def\clB{{\cal B}}
\def\clC{{\cal C}}

\def\clE{{\cal E}}

\def\clG{{\cal G}}

\def\clI{{\cal I}}
\def\clK{{\cal K}}
\def\clL{{\cal L}}

\def\clN{{\cal N}}

\def\clP{{\cal P}}

\def\clR{{\cal R}}
\def\clS{{\cal S}}

\def\clU{{\cal U}}

\def\clW{{\cal W}}
\def\clX{{\cal X}}

 \def\FRAC#1#2#3{\genfrac{}{}{}{#1}{#2}{#3}}

\def\ddtp{{\mathchoice{\FRAC{1}{d^{\hbox to 2pt{\rm\tiny +\hss}}}{dt}}%
{\FRAC{1}{d^{\hbox to 2pt{\rm\tiny +\hss}}}{dt}}%
{\FRAC{3}{d^{\hbox to 2pt{\rm\tiny +\hss}}}{dt}}%
{\FRAC{3}{d^{\hbox to 2pt{\rm\tiny +\hss}}}{dt}}}}

\def\average#1,#2,{{1\over #2} \sum_{#1}^{#2}}

\def\eye(#1){{\bf(#1)}\quad}

\newtheorem{theorem}{{\bf Theorem}}[section]

\newtheorem{proposition}[theorem]{{\bf Proposition}}
\newtheorem{lemma}[theorem]{{\bf Lemma}}

\def\eq#1/{(\ref{e:#1})}

\newcommand{\inp}[2]{{\langle #1, #2 \rangle}}
\newcommand{\inpr}[2]{{\langle #1, #2 \rangle}_\bR}

\newcommand{\beqn}[1]{\notes{#1}%
\begin{eqnarray} \elabel{#1}}

\newcommand{\eeqn}{\end{eqnarray} }

\newcommand{\beq}[1]{\notes{#1}%
\begin{equation}\elabel{#1}}

\newcommand{\eeq}{\end{equation}}

\def\bdes{\begin{description}}
\def\edes{\end{description}}

\newcounter{rmnum}

\newcounter{anum}

{\end{list}}

\def\ass(#1:#2){(#1\ref{#1:#2})}

\def\ritem#1{
\item[{\sf \ass(\current_model:#1)}]
}

\newenvironment{recall-ass}[1]{%
\begin{description}
\def\current_model{#1}}{
\end{description}
}

\usepackage{tikz}
\usepackage{pgfplots,tikz-3dplot}
\pgfplotsset{compat=newest}
\usetikzlibrary{patterns}
\usetikzlibrary{positioning}
\usetikzlibrary{datavisualization}
\usetikzlibrary{datavisualization.formats.functions}
\usetikzlibrary{backgrounds}
\usetikzlibrary{shapes,snakes}
\usepackage{scalefnt}

\renewcommand\thesection{\arabic{section}}
\renewcommand\thesubsection{\thesection.\arabic{subsection}}
\renewcommand\thesubsubsection{\thesubsection.\arabic{subsubsection}}

\renewcommand\thesectiondis{\arabic{section}}
\renewcommand\thesubsectiondis{\thesectiondis.\arabic{subsection}}
\renewcommand\thesubsubsectiondis{\thesubsectiondis.\arabic{subsubsection}}

\def\herm{{\sfH}}

\def\snr{{\mathsf{snr}}}

\newcommand{\mre}[1]{\Re\hspace{-1mm}\left[#1\right]}

\def\cg{{\clC\clN}}

\def\vol{{\mathsf{vol}}}

\def\RF{{\clX^n_{\mathsf{RF}}}}

\def\OC{{\clX^n_{\mathsf{OC}}}}
\def\OOC{{\overline{\clX^n_{\mathsf{OC}}}}}

\def\si{{\mathbf{s}}}
\def\cs{{\mathbf{c}}}
\def\conv{{\mathsf{co}}}
\def\closure{{\mathsf{cl}}}

\def\dsum{{\mathsf{sud}}}
\def\clSu{{\clS_+^\circ}}

\long\def\comment#1{}

\newfont{\bbb}{msbm10 scaled 700}

\newfont{\bb}{msbm10 scaled 1100}

\newcommand{\Sigmam}{\hbox{\boldmath$\Sigma$}}

\newcommand{\trace}{{\hbox{tr}}}

\renewcommand{\Re}{{\rm Re}}
\renewcommand{\Im}{{\rm Im}}

\newcommand{\transp}{{\sf T}}

\AtBeginDocument{\setlength\abovedisplayskip{4pt}}
\AtBeginDocument{\setlength\belowdisplayskip{4pt}}

\usepackage{afterpage}
\usepackage{balance}

\begin{document}

\title{\vspace{-0.2cm} Capacity and Degree-of-Freedom of OFDM Channels with Amplitude Constraint}
\author{\IEEEauthorblockN{Saeid Haghighatshoar\IEEEauthorrefmark{1},
Peter Jung\IEEEauthorrefmark{1}, Giuseppe Caire\IEEEauthorrefmark{1}}\\
\IEEEauthorblockA{\IEEEauthorrefmark{1}\IEEEauthorblockA{Communications and Information Theory Group, Technische Universit\"{a}t Berlin}}\\
Emails: saeid.haghighatshoar@tu-berlin.de, peter.jung@tu-berlin.de, caire@tu-berlin.de}

\maketitle

\begin{abstract}
In this paper, we study the capacity and degree-of-freedom (DoF) scaling for the continuous-time amplitude limited AWGN channels in radio frequency (RF) and intensity modulated optical communication (OC) channels. More precisely, we study how the capacity varies in terms of the OFDM block transmission time $T$, bandwidth $W$, amplitude $A$ and the noise spectral density $\frac{N_0}{2}$. We first find suitable discrete encoding spaces for both cases, and prove that they are convex sets that have a semi-definite programming (SDP) representation. Using tools from convex geometry, we find lower and upper bounds on the volume of these encoding sets, which we exploit to drive pretty sharp lower and upper bounds on the capacity. We also study a practical Tone-Reservation (TR) encoding algorithm and prove that its performance can be characterized by the statistical width of an appropriate convex set. Recently, it has been observed that in high-dimensional estimation problems under constraints such as those arisen in Compressed Sensing (CS) statistical width plays a crucial role. We discuss some of the implications of the resulting statistical width on the performance of the TR. We also provide numerical simulations to validate these observations.
\end{abstract}

\section{Introduction}\label{inro}
Shannon in his seminal paper \cite{shannon2001mathematical} derived the capacity per unit-time of the \textit{continuous-time additive white Gaussian noise} channel (CTAWGN) under input power constraint. The capacity is given by $W\log_2(1+\snr)$, where $W$ is the  bandwidth, and where $\snr=\frac{P}{N_0 W}$; $P$ denotes the input power and $\frac{N_0}{2}$ denotes the power spectral density of the white Gaussian  noise. He used the sampling theorem for the band-limited signals and the capacity formula $\frac{1}{2}\log_2(1+\frac{P}{\sigma^2})$ for the discrete-time AWGN (DTAWGN) under the input power constraint $P$ and noise variance $\sigma^2$. 
A more rigorous proof was later given by Slepian \cite{slepian1983some}, who introduced the Prolate spheroidal wave functions and proved the well-known $2WT$-result for the dimension of the signal space \textit{essentially} time-limited to $T$ and \textit{essentially} band-limited to $W$. 
In brief, the result states that for moderately large \textit{signal-to-noise ratio} (SNR), the capacity scales like $WT \log_2(\snr)$, and the \textit{degree-of-freedom} (DoF) for sufficiently high SNR is given by $WT$. 

Although the power-constraint is suitable for theoretical analysis and signal design (due to isometry of the underlying Hilbert spaces), in many practical scenarios in communications systems the front end of the CTAWGN channel is a power amplifier with a limited amplitude $A$. Thus, it is important to know how the capacity varies under the amplitude constraint. This requires studying the DoF of band-limited signals in $\bL_\infty[0,T]$ rather than $\bL_2[0,T]$.

For the discrete-time case, the amplitude-limited variant of the problem was studied by Smith \cite{smith1971information}, who also showed that the capacity achieving input distribution is discrete and has a finite support. 
The result was extended to the complex-valued Gaussian channels in \cite{shamai1995capacity}, and bounds on the capacity were derived in \cite{mckellips2004simple}. Recently, tighter bounds were obtained in \cite{thangaraj2015capacity} via a dual capacity result stated in \cite{csiszar2011information}. 
An interesting lower-bound was given in \cite{farrell2009kashin} by using a recent result of Lyubarskii and Vershynin \cite{lyubarskii2010uncertainty} on the existence of tight frames, which in turn uses a deep result of Kashin on the comparison of diameter of certain subsets of Banach spaces under different norms \cite{kashin1977diameters}.  In brief, the main idea in \cite{farrell2009kashin} is to admit some rate loss to transform codewords designed for the power-limited DTAWGN channel (codewords with limited $\ell_2$-norm) into codewords with limited amplitude (limited $\ell_\infty$-norm) suitable for amplitude-limited DTAWGN. 
However, in contrast with the power-limited case, where the discrete-time and the continuous-time variants are related via Hilbert space isometries, the results in \cite{farrell2009kashin} does not directly extend to the continuous-time case due to the lack of obvious isometry between $\ell_\infty(\bZ)$ and $\bL_\infty[0,T]$.

In this paper, we study the continuous-time variant of the problem when the transmitter uses an orthonormal collection of waveforms $\clW^n=\{\psi_k(t)\}_{k=1}^n$, for signal modulation. The $n$-dim encoding space for $\clW^n$, called $\clX^n\subset \bC^n$, consists of all coefficients $\bfx=(x_1,x_2, \dots,x_n)\in \bC^n$ for which the transmitted waveform $\sum_{k \in [n]} x_k \psi_k(t)$ has a limited amplitude $A$ for all $t\in[0,T]$, where $[n]$ denotes the set of integers $\{1,2,\dots,n\}$. We focus on the signal space of harmonic waveforms  $\clW^n_\sfH=\{e^{jk \frac{2\pi}{T} t}\}_{k=1}^n$, known as \textit{orthogonal frequency division modulation} (OFDM). 
It has been vastly studied in the literature due to its simple implementation (with FFT algorithm), robustness to multi-path fading in wireless scattering channels, and its information theoretic optimality for water-filling type encoding under frequency-selective Gaussian channels \cite{nee2000ofdm}. 

Coding for amplitude-limited OFDM channels is related to the well-known \textit{peak-to-average power ratio} (PAPR) reduction problem. In brief, as the communication is over the AWGN channel, the performance is an increasing function of the average transmit power. However, due to the amplitude limitation $A$, the peak power is limited. Therefore, the PAPR is a measure of the efficiency of the signal set (code in the signal space), and its minimization corresponds to maximizing the average transmit power. Several techniques have been proposed to tackle this problem such as coding, tone reservation, amplitude clipping, clipping and filtering, tone injection and partial transmit sequences. A summary of important results can be found in the survey papers \cite{wunder2013papr, han2005overview}. 
There is another collection of work on designing good codes with reasonable PAPR as in the papers \cite{davis1997peak, ochiai1997block, van1996ofdm, paterson2000generalized} (see also \cite{paterson2000existence} and the refs. therein). While these codes reduce PAPR, they also reduce the transmission rate severely, especially  for very large number of subcarriers.
%

\noindent{\bf Contribution.}
In this paper, we make a connection between the  OFDM for traditional \textit{radio frequency} (RF) channels and OFDM for optical intensity channels (OC) under the intensity constraint $A$ (see \cite{you2002upper, hranilovic2004capacity} and references therein). 
We show that the encoding set for both cases has a \textit{semi-definite programming} (SDP) representation, which we use to study the capacity scaling performance for both channels. In particular, we show that even though OFDM achieves the optimal DoF scaling $WT$ for OC in high SNR, it seems that, the best possible DoF in moderate SNR is at most $\lambda W T$, with a multiplicative loss $\lambda \in (0,1)$. Interestingly, this confirms a recent result of \cite{ilic2009papr}, where the authors using the results of Lyubarskii and Vershynin in \cite{lyubarskii2010uncertainty}, show (and numerically verify) that if the columns of a subsampled Fourier matrix build a tight frame, then it is possible to make all codewords have a constant PAPR at the cost of a multiplicative rate loss due to those carriers reserved for shaping (\textit{Tone-Reservation} (TR)). We study this problem further, and provide additional evidence that the constant PAPR seems to be achievable, by formulating it as the statistical width of a specific convex set. Our results suggest that, in OFDM systems with a large number of subcarriers, TR is a promising approach in order to achieve an effective PAPR reduction at the cost of a fixed multiplicative loss of DoFs.

\section{Statement of the Problem}

%
\subsection{Basic Setup}
Let $\bfx=(x_1, x_2, \dots, x_n) \in \bC^n$ be a sequence of symbols of length $n$ to be transmitted across a CTAWGN channel. In the OFDM modulation with $n$ subcarriers, the sequence $\bfx$ is transformed to a base-band signal given by
\begin{align}\label{eq:bb_sig}
s_b(t)=\sum_{k\in [n]} x_k e^{j2k\pi  t}, \ t \in [0,1).
\end{align}
Here, for convenience, we consider a normalized signal set with a normalized transmission time $T = 1$ and amplitude $A=1$. We will re-scale the results at the end to make the main system parameters $W$, $T$ and $A$ and $N_0$ appear explicitly. 
The harmonic wave-forms $\clW_\sfH=\{e^{j2k\pi t}\}_{k=1}^n$ in \eqref{eq:bb_sig} form an orthonormal collection with minimum frequency separation $1$ over $[0,1)$, under the inner product $\inp{\alpha(t)}{\beta(t)}=\int_{0}^1 \alpha(t)^\herm \beta(t) dt$.

We focus on two main applications of OFDM: 1) in passband modulation for RF channels, and 2) in base-band modulation for OC over intensity modulated channels. In the former, the baseband signal $s_b(t)$ is heterodyned to a sufficiently high \textit{normalized} carrier  frequency $f_c \gg  n$, and the resulted real-valued passband signal $s(t)=\Re[s_b(t) e^{j 2\pi f_c t}]$ is transmitted via the antenna (see Fig.~\ref{fig:ofdm_rad_opt_and_ctawgn}). The passband signal $s(t)$ occupies a normalized one-side bandwidth $W=n$ (or more precisely $n+1$) for sufficiently large block-length $n$. For OC, the real-valued signal $s(t)=1+\Re[s_b(t)]$ is used to modulate the light intensity of an optical diode, where we assume that, due to the unipolar nature of the diode, an additional normalized d.c. level of size $1$ is added to obtain a positive signal $s(t)$ (see Fig.~\ref{fig:ofdm_rad_opt_and_ctawgn}).
\begin{figure}[h]
\centering
\resizebox{8.5cm}{!}{

\begin{tikzpicture}[
block/.style ={rectangle, draw=blue, fill=white, thick, align=center, rounded corners, minimum height=2em},
block2/.style ={rectangle, fill=white, align=center, minimum height=2em}
]
\pgfmathsetmacro{\scale}{0.85}
\pgfmathsetmacro{\xshift}{7cm}
\pgfmathsetmacro{\yshift}{0.05cm}
\pgfmathsetmacro{\vshift}{3}

\begin{scope}[scale=\scale]
\node[block2] at (-1,0) (x_seq) {$x_n,x_{n-1}, \dots, x_1$};
\node[block] at (2,0) (mod) {Modulator};
\node at (4.5,0) (hdyn) {{\large$\otimes$}};
\node[block2] at (4.5,-1) (carr) {$e^{j 2\pi f_c t}$};
\node[block] at (6,0) (re) {$\Re[.]$};

\draw[->,thick] (x_seq) -- (mod);
\draw[->, thick] (mod) --  (hdyn) node[midway, above] {$s_b(t)$};
\draw[->, thick] (carr) -- (hdyn);
\draw[thick, ->] (hdyn) -- (re);
\draw[thick] (re) -- (7,0) -- (7,0.8);
\draw[thick] (7,0.8) -- (7.2,1);
\draw[thick] (7,0.8) -- (6.8,1) node[above, left]{Antenna};
\end{scope}

\begin{scope}[scale=\scale, yshift=-\vshift cm]
\node[block2] at (-1,0) (x_seq) {$x_n,x_{n-1}, \dots, x_1$};
\node[block] at (2,0) (mod) {Modulator};
\node[block] at (4.5,0) (re) {$\Re[.]$};
\node[block2] at (6,0) (dc) {\large $\oplus$};

\draw (x_seq) -- (mod);
\draw[thick, ->] (mod) -- (re) node[midway, above] {$s_b(t)$};
\draw[thick,->] (re) -- (dc);
\draw[thick,->] (6,-1) node[below] {d.c. level} -- (dc);
\draw[thick] (dc) -- (7,0) -- (7,1);

\draw[fill] (7,1-0.2)--(7.15,1-0.2) -- (7,1.2-0.2) -- (6.85,1-0.2) -- (7,1-0.2);
\draw[thick] (6.85,1.2-0.2) -- (7.15,1.2-0.2);
\draw[thick] (7,1.2-0.2) -- (7,1.5-0.2);
\draw (6.7,1.5-0.2) -- (7.3,1.5-0.2);
\draw (6.85,1.6-0.2) -- (7.15,1.6-0.2);
\draw (6.75,1) node [left] {Photo diode};
\end{scope}

\begin{scope}[xshift=\xshift, yshift=\yshift]
\draw[->, thick] (0,0) node[fill=white] {$x_1$} -- (1.8,0);
\draw (2,0) circle (0.2) node{$+$};
\draw[->, thick] (2,1) node[fill=white] {$w_1 \sim \cg(0, N_0)$} -- (2,0.2);
\draw[->, thick] (2.2,0) -- (4,0) node[right] {$y_1$};

\node at (0,-1) {$\vdots$};
\node at (2,-1) {$\vdots$};
\node at (4,-1) {$\vdots$};

\draw[->, thick] (0,0-\vshift) node[fill=white] {$x_{n}$} -- (1.8,0-\vshift);
\draw (2,0-\vshift) circle (0.2) node{$+$};
\draw[->, thick] (2,1-\vshift) node[fill=white] {$w_{n} \sim \cg(0, N_0)$} -- (2,0.2-\vshift);
\draw[->, thick] (2.2,0-\vshift) -- (4,0-\vshift) node[right] {$y_{n}$};
\end{scope}

\end{tikzpicture}
}
\caption{OFDM for RF and OC channels, and the equivalent discrete-time model.}
\label{fig:ofdm_rad_opt_and_ctawgn}
\end{figure}
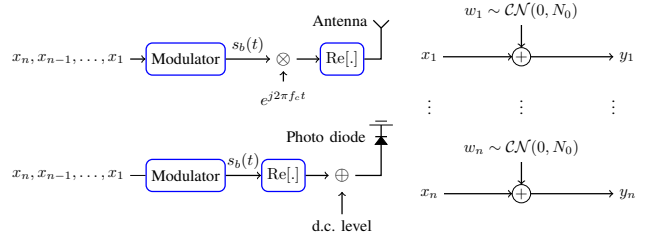
In both cases, the OFDM symbol $\bfx$ is detected at the receiver via an electronic circuit and matched filtering. The equivalent discrete-time channel can be modelled as parallel independent complex-valued Gaussian channels. This has been shown in Fig.~\ref{fig:ofdm_rad_opt_and_ctawgn}, where $\frac{N_0}{2}$ denotes the power spectral density of the channel.

\subsection{Behaviour for Large Number of Subcarriers}
In this paper, we are interested in the regime where $n$ is very large. Although in the ergodic regime, coding across many consecutive OFDM blocks can be used to further boost the performance, here we mainly focus on the \textit{one-shot} behaviour of the system where an individual OFDM block carries a huge number of symbols. This may be motivated by the recent LTE standards in mobile communications. We are mainly interested to know how the one-shot achievable rate per unit time scales in terms of $W$, $T$, $A$ and noise parameter $N_0$.

\section{Encoding Space}
Let $\bfa: [0,1) \to \bC^n$ be the vector-valued function given by 
\begin{align}\label{a_def}
\bfa(t)=[e^{-j2\pi t}, e^{-j4\pi t}, \dots, e^{-j2n\pi t}]^\transp, \ t \in [0,1).
\end{align}
The baseband signal \eqref{eq:bb_sig} can be written as $s_b(t)=\bfa(t)^\herm \bfx$. We define the encoding space for RF and OC as
\begin{align}
\RF&=\{\bfx \in \bC^n: |\bfa(t)^\herm \bfx|\leq 1\text{ for } t\in [0,1)\},\label{rf_seq}\\
\OC&=\{\bfx \in \bC^n: 1+\Re[\bfa(t)^\herm \bfx] \geq 0\text{ for } t\in [0,1)\},\label{oc_seq}
\end{align}
where we assume that the amplitude of the RF signal and the d.c. level of the OC signal are both normalized to $1$. 
A code of rate $R$ for communication over AWGN via OFDM signalling is a mapping $\scrC_{\mathsf{RF}} : [2^{nR}] \to \RF$ and $\scrC_{\mathsf{OC}} : [2^{nR}] \to \RF$. 
Both $\RF$ and $\OC$ are not polyhedral, however, they can be  represented as \textit{Linear Matrix Inequalities} (LMI) over the cone of \textit{positive semi-definite} (PSD) matrices, and they have \textit{semi-definite programming} (SDP) representation.
We first consider $\RF$. As a corollary to Theorem 4.24 in \cite{dumitrescu2007positive}, we have:
\begin{proposition}\label{rf_mem}
Let $\bfx \in \bC^n$. Then, $\bfx \in \RF$ if and only if there exists a Hermitian matrix $\bfH \in \bC^{n\times n}$ such that 
\begin{align}\label{eq:psd_char}
\left [\begin{matrix} \bfH & \bfx \\ \bfx^\herm & 1 \end{matrix}\right ] \succeq 0, \sum_{\ell=1}^{n-i} \bfH_{\ell,\ell+i}=\left\{ \begin{array}{ll} 1, & i=0, \\ 0, & i=1, \dots, n-1. \end{array} \right.
\end{align}
\end{proposition}

\begin{proof}
We only prove one direction by simply applying the Schur's decomposition to the  PSD matrix stated in the SDP constraint. Thus, we obtain $\bfH - \bfx \bfx^\herm \succeq 0$, which implies that for every $\bfw \in \bC^n$, we have $\bfw^\herm \bfH \bfw \geq |\bfw^\herm \bfx|^2$. Fixing an arbitrary  $t \in [0,1)$ and setting $\bfw=\bfa(t)\in \bC^n$, we obtain that $\bfw^\herm \bfH \bfw=1$, and as a result $|\bfa(t)^\herm \bfx| \leq 1$. Since this is true for every $t \in [0,1)$, the proof is complete.
\end{proof}

We can obtain a similar representation for the set $\OC$. We need some notations first. We define $\clS_+$ as the set of all $(n+1)\times (n+1)$ Hermitian PSD matrices, and $\clSu=\{\bfX \in \clS_+: \trace[\bfX]=1\}$ as the affine subspace of $\clS_+$ with unit trace. Note that $\clSu$ is a closed convex subset of $\clS_+$. Moreover, for any $\bfX\in \clSu$, we have $ \trace[\bfX]=\sum_{i\in[n+1]} \lambda_i=1$, where $\lambda_i\geq 0$ denote the eigen-values of $\bfX$.  This implies that $\clSu$ is also bounded. We define $\dsum: \clSu \to \bC^n$ as the sum-diagonal map, which for  every $\bfX \in  \clSu$ and for $k \in [n]$, gives a vector $\dsum(\bfX)$ by
\begin{align}\label{dsum_def}
[\dsum(\bfX)]_k=2\,\diag(\bfX,-k):=\sum_{\ell=k}^{n+1} 2\,\bfX_{\ell,\ell-k+1}.
\end{align}
As a corollary to Theorem 3 in \cite{davidson2002linear}, we have the following representation of $\OC$.
\begin{proposition}\label{oc_mem}
Let $\bfx \in \bC^n$. Then, $\bfx \in \OC$ if and only if there is an $\bfX \in \clSu$ such that $\bfx=\dsum(\bfX)$. \hfill $\square$
\end{proposition}
\begin{proof}
One side is again easy to prove. Let $t\in[0,1)$ and let $\bfb(t)=[1, e^{j2\pi t}, e^{j 4\pi t}, \dots, e^{j2n\pi t}]^\transp$. It is not difficult to see that $\bfb(t)\bfb(t)^\herm$ is a Toeplitz matrix with $1$ on its main diagonal and with $e^{j2k\pi t}$ and $e^{-j2k\pi t}$ on its $k$-th lower and $k$-th upper diagonal respectively. Now let $\bfx \in \bC^n$ and suppose there is an $\bfX \in \clSu$, with $\bfx=\dsum(\bfX)$. This implies that 
\begin{align}
0 &\leq \bfb(t)^\herm \bfX \bfb(t)=\trace[\bfX\,\bfb(t)\bfb(t)^\herm]\\
&=1+ \sum_{k\in [n]} \diag(\bfX,-k) e^{j 2k\pi t} + \diag(\bfX,k) e^{-j 2k\pi t}\\
&=1+ \Re[\bfa(t)^\herm \dsum(\bfX)]=1+\Re[\bfa(t)^\herm \bfx].
\end{align}
Since this is true for every $t \in [0,1)$, then $\bfx \in \OC$. 
\end{proof}

The next proposition shows some of the properties and also the relation between $\RF$ and $\OC$. 
\begin{proposition}\label{upper_bound}
The sets $\RF$ and $\OC$ are compact and convex subsets of $\bC^n$. Moreover, $\RF=\cap_{\phi \in [0, 2\pi)} \{e^{j \phi} \OC\}$. In particular, $\RF \subset \OC$. \hfill $\square$
\end{proposition}
\begin{proof}
The compactness and the convexity can be directly checked from the definition. However, it is also seen from Proposition \ref{rf_mem} and \ref{oc_mem}, that $\RF$ and $\OC$ are obtained from linear projection of compact and convex subsets of $(n+1)\times (n+1)$ PSD matrices, thus, they must be compact and convex. 

To prove the next part, let $\clL=\cap_{\phi \in [0, 2\pi)}\{e^{j \phi} \OC\}$. Note that $\clL$ is a symmetric set, i.e., $e^{j \phi} \clL=\clL$ for all $\phi \in [0,2\pi)$. It is not also difficult to see that it is the largest symmetric set contained in $\OC$. We use this property to prove the theorem.

First note that since $\RF$ is itself a symmetric set, it must be contained in $\clL$, i.e., $\RF \subset \clL$. 
To prove the other direction, let $\bfx \in \clL$ and let $t \in [0,1)$. From the symmetry of $\clL$, it results that for every $\phi\in [0,2 \pi)$, and in particular, $\phi=\pi-\angle\bfa(t)^\herm \bfx$, the vector $e^{j \phi}\bfx$ belongs to $\clL$, thus, to $\OC$. Hence, we must have $1+\mre{e^{j \phi}\bfa(t)^\herm \bfx} \geq 0$, which implies that $|\bfa(t)^\herm \bfx|\leq 1$. Since this is true for an arbitrary $t \in [0,1)$, we obtain $\clL \subset \RF$. 
This completes the proof.
\end{proof}

\section{Lower and Upper Bound on the Volume}\label{vol_section}
In this section, we drive lower and upper bounds on the volume of $\RF$ and $\OC$ via results in convex geometry. 

\subsection{Real Embedding} 
Since the results can be stated more conveniently in the real-valued case we first embed both sets $\RF$ and $\OC$ in $\bR^{n\times 2}$. We identify every $\bfx \in \bC^n$ by the $n\times 2$ matrix $X=\scrR(\bfx)=[\Re[\bfx], \Im[\bfx]]$, where $\scrR: \bC^n \to \bR^{n\times 2}$ denotes this real embedding. We denote the inverse map by $\scrI: \bR^{n\times 2}\to \bC^n$. It is not difficult to check the isometry 
\begin{align}
\inp{\scrR(\bfx_1)}{\scrR(\bfx_2)}=\inpr{\bfx_1}{\bfx_2}=\Re[\bfx_1^\herm \bfx_2],
\end{align} 
where the first inner product is the conventional inner product between matrices in $\bR^{n\times 2}$ given by $\inp{X_1}{X_2}=\trace[X_1^\transp X_2]$. We define $A(t)=\scrR(\bfa(-t))$, where $\bfa(t)$ is given by \eqref{a_def}. We can check that the columns of $A(t)$ are given by the vector  $\cs(t):=[\cos(2\pi t), \cos(4\pi t), \dots, \cos(2n\pi t)]^\transp$ and $\si(t):=[\sin(2\pi t), \sin(4\pi t), \dots, \sin(2n\pi t)]^\transp$. We also define the matrix $B(t):=[-\si(t),\,\cs(t)]$. We can also see that multiplying the vector $\bfa(t)$ by the constant phase $e^{j \phi}$ is equivalent to rotating the columns of $A(t)$ by a one parameter rotation group. More precisely, we have
\begin{align}\label{ephi_rot}
\scrR(e^{j \phi} \bfa(-t))= \cos(\phi) A(t)+ \sin(\phi) B(t).
\end{align}

\subsection{Polar of a set} 
For a subset $\clC$ of $\bR^{n\times 2}$, we define the polar (symmetric polar) of $\clC$ as 
\begin{align}\label{polar_def}
\clC^\circ=\{Y\in \bR^{n\times 2}: \sup_{X\in \clC} |\inp{X}{Y}| \leq 1\}.
\end{align} 
The polar is typically defined without the absolute value but since we always work with symmetric sets, we keep the absolute value. It is not difficult to see from \eqref{polar_def} that $\clC^{\circ}$ is always closed and convex, and contains the origin. From duality in convex geometry, it results that if $\clC$ is convex, closed and symmetric then $\clC^{\circ \circ}=\clC$. From \eqref{polar_def} it is not also difficult to see that the polar set does not change if we replace $\clC$ by $\closure\{ \conv\{\clC, - \clC\}\}$, where $\closure$ and $\conv$ denote the closure and convex hull operation. This implies that for every set $\clC$, we have $\clC^{\circ \circ}=\closure\{ \conv\{\clC, - \clC\}\}$. Another property that we will use is that if $\clC \subset \clB$ then $\clB^{\circ} \subset \clC^{\circ}$, and $(\lambda\clC)^{\circ}= \frac{1}{\lambda} \clC^{\circ}$.

\subsection{Mahler and Bottleneck Conjecture}
Let $\clK\subset \bR^{n\times 2}$ be a symmetric convex set with a polar set $\clK^{\circ}$. The Mahler volume of $\clK$ is defined as $\nu(\clK)=\vol(\clK)\vol(\clK^{\circ})$, where $\vol$ denotes the volume. Mahler in \cite{mahler1939ubertragungsprinzip} conjectured that for every such $\clK$
\begin{align}
\frac{4^{2n}}{(2n)!} \leq \nu(\clK) \leq \frac{\pi^{2n}}{(n!)^2},
\end{align}
where the upper and the lower bound are achieved for the sphere and the cube respectively (note that in our case the dimension is $2n$, and the conjecture is stated for $2n$).
The upper bound was proven by Santal\'o \cite{santalo2009invariante}, and is known as the Blaschke-Santal\'o inequality. The conjecture for the lower bound is still open but a weak variant of it, known as the \textit{Bottleneck conjecture}, has been proven which implies the Mahler conjecture up to a factor of $(\frac{\pi}{4})^{2n} \gamma_n$, where $\gamma_n$ is monotonic factor that begins at $\frac{4}{\pi}$ and increases to $\sqrt{2}$ as the $n$ goes to infinity \cite{kuperberg2008mahler}. In our case, it gives the lower bound  $\nu(\clK)\geq \sqrt{2} \frac{\pi^{2n}}{(2n)!}$ for any symmetric convex set $\clK$.

\subsection{Lower bound on $\vol(\OC)$} 
Let $\OC$ be encoding set defined by \eqref{oc_seq}. We define the symmetric part of $\OC$ as $\OOC=(-\OC \cap \OC)$, where it is easy to see that
\begin{align*}
\OOC=\{\bfx \in \bC^n: -1\leq \Re[\bfa(t)^\herm \bfx] \leq 1 \text{ for } t\in [0,1)\}.
\end{align*}
Let $\clK_1 \subset \bR^{n\times 2}$ be the real embedding of $\OOC$. It is not difficult to see that  $\clK_1$ can be identified with the polar of the set $\clA_\pm=\cup_{t \in [0,1)}\{A(t), -A(t)\}$. In particular, we have that $\clK_1^\circ=\closure( \conv (\clA_\pm))$. We also define $\clA_+=\cup_{t \in [0,1)}\{A(t)\}$, $\clA=\conv(\clA_+)$, and $\clE=\scrI(\clA)$. It is not difficult to check that the convexity remains invariant under the real embedding $\scrR$ and its inverse $\scrI$. This implies that $\clE=\conv(\cup_{t\in[0,1)}\{\bfa(-t)\})=\conv(\cup_{t\in[0,1)}\{\bfa(t)\})$ due to the periodicity of $\bfa(t)$ in $[0,1)$. 

We first prove that $\clK_1^\circ \subset (\clA-\clA)$, where $(\clA-\clA):=\{X-Y: X,Y \in \clA\}$ denotes the Minkowski difference of $\clA$. We need the following preliminary lemma.
\begin{lemma}\label{origin_inside}
The set $\clA$ contains the origin. 
\end{lemma}
\begin{proof}
Note that $A(t)\in \clA$ for every $t \in[0,1)$. Let $\mu$ be the uniform probability measure over $[0,1)$. Since $\clA$ is convex, it results that ${\bf 0}=\int_0^1 A(t) \mu(dt) \in \clA$, thus, $\clA$ contains the origin. 
\end{proof}

\begin{proposition}\label{K1_A-A}
Let $\clA$ and $\clK_1^\circ$ be as defined before. Then $\clK_1^\circ\subset (\clA - \clA)$. \hfill $\square$
\end{proposition}
\begin{proof}
From Lemma \ref{origin_inside}, it  results that the convex set $\clA$ contains the origin. This implies that $\clA, -\clA \subset (\clA -\clA)$, and especially $\clA_\pm \subset (\clA -\clA)$. Since $(\clA-\clA)$ is convex, it  results that $\clK_1^\circ=\closure( \conv (\clA_\pm))\subset(\clA-\clA)$. 
\end{proof}

Consider the set $\clE$ defined before, where it is seen that $\clE$ is the convex hull of the periodic curve $\cup_{t\in[0,1)}\{\bfa(t)\}$. Sch\"onberg in \cite{schoenberg1954isoperimetric} using isoperimetric inequalities proved that $\vol(\clE)$ is strictly larger than the volume of convex hull of any other periodic curve with the same length, where the equality holds if and only if the curve is obtained by some linear isometry from the curve $\cup_{t\in[0,1)}\{\bfa(t)\}$. He also showed that $\vol(\clE)=\frac{2^n \pi^n n!}{(2n)!}$. 
Thus, we obtain that $\vol(\clA)=\vol(\clE)=\frac{2^n \pi^n n!}{(2n)!}$. 
Moreover, from Rogers-Shephard inequality \cite{rogers1957difference}, $\vol(\clA-\clA)$ can be upper bounded by ${4n \choose 2n} \vol(\clA)$. Thus, using Proposition \ref{K1_A-A}, we obtain 
\begin{align}\label{vk_circ_upperbound}
\vol(\clK_1^{\circ}) \leq \vol(\clA-\clA) \leq {4n \choose 2n} \frac{2^n \pi^n n!}{(2n)!}.
\end{align}
Since $\clK_1$ and $\clK_1^\circ$ are symmetric and convex, applying the result of Bottleneck conjecture, i.e., $\nu(\clK_1)=\vol(\clK_1)\vol(\clK_1^\circ)\geq \sqrt{2} \frac{\pi^{2n}}{(2n)!}$, and using \eqref{vk_circ_upperbound}, we obtain that 
\begin{align}\label{oc_lb}
\vol(\OC) \geq \vol(\clK_1)\geq \frac{\sqrt{2} \pi^n}{2^n n! {4n \choose 2n}}.
\end{align}

\subsection{Lower Bound on $\vol(\RF)$}
Let $\clK_2$ be the real embedding of $\RF$. Note that from the definition of $\RF$ in  \eqref{oc_seq} and the one-dimensional rotation property mentioned in \eqref{ephi_rot}, we can see that $\clK_2$ is given by the polar of the parametric set 
\begin{align}\label{a_phi}
\clA_\phi:=\cup_{t\in [0,1), \phi\in [0, 2\pi)}\{\cos(\phi) A(t)+ \sin(\phi) B(t)\},
\end{align}
where in the special case $\phi \in\{0,\pi\}$, we obtain $\clA_\pm=\cup_{t \in [0,1)}\{A(t), -A(t)\}$. Let $\clB_\pm=\cup_{t \in [0,1)}\{B(t), -B(t)\}$. It is not difficult to see that finding a lower bound for $\vol(\clK_2)$ at least requires finding an upper bound on the volume of the convex hull of $\clA_\pm \cup \clB_\pm$. This seems quite challenging especially that, in contrary to Rogers-Shephard inequality, there is no universal upper bound on the volume of the convex hull of the union of two different sets in terms of volume of their individual convex hulls.
Instead, we use another approach to find a lower bound on $\vol(\RF)$. 
We consider two grid of size $n$ on $[0,1)$ given by $\clG_1=\{\frac{i-1}{n}: i \in [n]\}$ and $\clG_2=\{\frac{2i-1}{2n}: i \in [n]\}$. We set $\bfG_1$ as a matrix whose columns are given by $\bfa(g)$, $g \in \clG_1$. We define $\bfG_2$ similarly by using grid points from $\clG_2$. It is not difficult to see that $\bfG_1$ coincides with the DFT matrix. In particular, $\bfG_1^\herm \bfG_1=\bfG_2 ^\herm \bfG_2=n \bfI_n$. We also define $\bfK=\frac{1}{n} \bfG_2^\herm \bfG_1$. We can simply check that $\bfK$ is a Toeplitz matrix with $\bfK_{\ell m}=\sfK(\ell -m+1/2)$ where $\sfK(t)=e^{j \pi t} \frac{\sin(n\pi t)}{n\sin(\pi t)}$. Note that $\clG=\clG_1 \cup \clG_2$ is a uniform grid over $[0,1)$ of size $2n$, with an oversampling factor $2$. Let $\overline{\RF}$ be the approximation of $\RF$ via grid $\clG$, i.e., $\overline{\RF}=\{\bfx: \max_{g \in \clG} |\bfa(g)^\herm \bfx| \leq 1\}$.
In \cite{wunder2002peak}, it was shown that for an oversampling factor $\gamma>1$
\begin{align}
\sup_{t\in [0,1)} |\bfa(t)^\herm \bfx| \leq \frac{1}{\cos(\frac{\pi}{2\gamma})} \max_{g \in \clG} |\bfa(g)^\herm \bfx|.
\end{align}
In our case $\gamma=2$, and we obtain that $\RF \subset \overline{\RF} \subset \sqrt{2} \RF$, which implies that $\vol(\RF) \geq 2^{-n} \vol(\overline{\RF})$. 

Let  $\clU=\{\bfx\in \bC^n: \max_{k \in [n]} |x_k|\leq 1\}$ be the complex $\ell_\infty$-ball. Note that $\overline{\RF}=\clC_1\cap \clC_2$, where $\clC_i=\{\bfx: \bfG_i^\herm \bfx \in \clU \}, i\in[2]$. 
To lower bound $\vol(\overline{\RF})$, we need to find a lower bound for the intersection of two convex sets $\clC_1$ and $\clC_2$.  Since both $\clC_1$ and $\clC_2$ are symmetric sets containing the origin, we find a constant $\beta_n\geq 1$ such that
$\frac{1}{\beta_n} \clC_1 \subset \clC_2$, which implies that $\frac{1}{\beta_n} \clC_1 \subset \clC_1 \cap \clC_2$. This has been pictorially shown in Fig.~\ref{fig:ellipse}.
\begin{figure}[h]
\centering
\resizebox{4.5cm}{!}{
\begin{tikzpicture}[scale=0.8]
\pgfmathsetmacro{\vshift}{3}

\begin{scope}[rotate=20]
\draw[fill=gray!10, thick] (0,0) ellipse (3cm and 1cm);
\draw (3,1) node[right, below] {$\clC_1$};
\end{scope}

\begin{scope}[rotate=70]
\draw[fill=gray!10,thick] (0,0) ellipse (3cm and 1cm);
\draw (3,1) node[right, below] {$\clC_2$};
\end{scope}

\begin{scope}[rotate=20]
\clip (0,0) ellipse (3cm and 1cm);
\begin{scope}[rotate=50]
\draw[fill=gray!30] (0,0) ellipse (3cm and 1cm);
\end{scope}
\end{scope}

\begin{scope}[rotate=20]
\draw[thick, dashed] (0,0) ellipse (3cm and 1cm);
\end{scope}
\begin{scope}[rotate=70]
\draw[thick, dashed] (0,0) ellipse (3cm and 1cm);
\end{scope}

\begin{scope}[rotate=20]
\path[fill=gray!50, draw=red, thick] (0,0) ellipse (1.2cm and 0.4cm);
\draw[->, thick, blue] (1,-2).. controls (0.2,-1.5) ..  (0,-0.4);
\draw (1,-2)  node[right ] {$\frac{1}{\beta_n} \clC_1$};
\end{scope}

\draw[->, thick, blue] (-1,2.5).. controls (0,2) ..  (0,1.1);
\draw (-1,2.5) node[above, left] {$\clC_1\cap \clC_2$};
\end{tikzpicture}
}
\caption{Approximating $\clC_1 \cap \clC_2$ by a scaled version of $\clC_1$.}
\label{fig:ellipse}
\end{figure}
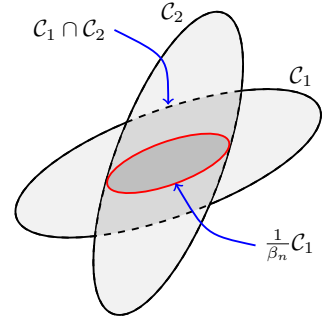
\begin{proposition}\label{beta_factor}
There is a parameter $\beta_n\geq 1$ such that $\frac{1}{\beta_n} \clC_1 \subset \clC_1 \cap \clC_2$. Moreover, $\beta_n\approx \theta \log_2(n)$ for $\theta \approx \frac{\log(2)}{\pi}$ for sufficiently large $n$. \hfill $\square$
\end{proposition}
\begin{proof}
From the definition of $\clC_1$ and $\clC_2$, it is seen that the required scaling factor $\beta_n$ is given by 
\begin{align}
\beta_n&=\max_{\bfx \in \clU} \|\bfG_2^\herm \bfG_1\bfx\|_\infty=\max_{\bfx \in \clU} \|\bfK \bfx\|_\infty\\
&\leq \sum_{\ell \in [n]} |\sfK(\ell+1/2)|\leq \sum_{\ell\in [n]} \frac{2}{\pi(2\ell+1)}\approx \theta \log_2(n),\nonumber
\end{align}
where $\theta\approx \frac{\log(2)}{\pi}$.  This completes the proof.
\end{proof}

\begin{proposition}\label{logn_volume}
For sufficiently large $n$, the volume of $\RF$ is lower bounded by $\frac{\pi^n}{2^n \theta^{2n} \log_2(n)^{2n} n^n}$. \hfill $\square$
\end{proposition}
\vspace{1mm}
\begin{proof}
First note that the volume of $\ell_\infty$-ball $\clU$ is given by $\pi^n$. Since $\frac{1}{\sqrt{n}} \bfG_i$ are unitary matrices, we have $\vol(\clC_i)=\frac{\vol(\clU)}{(\sqrt{n})^{2n}}=\frac{\pi^n}{n^n}$. Using Proposition \ref{beta_factor}, we obtain 
\begin{align}
\vol(\RF) \geq 2^{-n} \vol(\overline{\RF})\geq \frac{\vol(\clC_1)}{2^{n}\beta_n^{2n}}= \frac{\pi^n}{2^n \theta^{2n} \log_2(n)^{2n} n^n}, \nonumber
\end{align} 
which is the desired result.
\end{proof}

\subsection{Upper Bound on the Volume}
In \cite{you2002upper}, it was shown that the encoding set $\OC$ is contained in $\clE=\conv(\cup_{t\in [0,1)}\{\bfa(t)\})$ as defined before. This can be written as 
\begin{align}\label{E_set}
\clE=\big \{\bfx \in \bC^n:\,  &\exists\, \mu \text{ such that} \int_{0}^1 \bfa(t) \mu(dt)=\bfx\big \},
\end{align}
where $\mu$ denotes a probability measure.
This can be easily proved since if $\bfx \in \OC$, then $x(t)=1+\Re[\bfa(t)^\herm \bfx]$ is  positive  with $\int_0^1 x(t) dt=1$. And it can be easily checked that $\mu(dt)=x(t) dt$ as a probability measure over $[0,1)$ satisfies $ \int_{0}^1 \bfa(t) \mu(dt)=\bfx$, thus, $\bfx$ should belong to $\clE$. 
As we explained in Section \ref{vol_section}, the volume of $\clE$ was computed by Sch\"onberg in \cite{schoenberg1954isoperimetric} to be  $\vol(\clE)=\frac{2^n \pi^n n!}{(2n)!}$. Hence, from $\RF \subset \OC \subset \clE$, we have that 
\begin{align}\label{vol_up_bound}
\vol(\RF) \leq \vol(\OC) \leq \frac{2^n \pi^n n!}{(2n)!}.
\end{align}

\section{Capacity Results}\label{cap_results}
\subsection{Lower Bounds on the Capacity}
 We define $\clP_\RF=\{p: \supp(p) \subset \RF\}$, as the set of all probability measures supported on $\RF$. We define $\clP_\OC$ similarly.
Since $\RF \subset \OC$, we have $\clP_\RF \subset \clP_\OC$. 
Let $p\in \clP_\RF$, and let $\bfx \sim p$. We define the mutual information between the input and output of the channel by
\begin{align}
I_p(\bfx;\bfy)=h_p(\bfy)-h_p(\bfy|\bfx)=h_p(\bfy)-h(\bfw),
\end{align}
where $h_p(.)$ denotes the differential entropy under $\bfx \sim p$, and where $\bfw$ is the additive noise vector (see Fig.~\ref{fig:ofdm_rad_opt_and_ctawgn}). Note that since $\RF$ is a compact set, $\bfy=\bfx+\bfw$ has a well-defined covariance matrix $\Sigmam_y=\Sigmam_x+N_0 \bfI_n$, thus, 
\begin{align}
h_p(\bfy) \leq h(\bfy_g)=\log_2((\pi e)^n |\Sigmam_y|) <\infty,
\end{align}
where $\bfy_g$ is a complex Gaussian vector with the same covariance as $\Sigmam_y$. This implies that $\bfI_p(\bfx;\bfy)$ is always well-defined for every $p \in \clP_\RF$. We also define 
\begin{align}\label{capn}
C^{(n)}_\RF=\sup_{p\in \clP_\RF} I_p(\bfx;\bfy).
\end{align}
\begin{proposition}\label{cap_lower}
Let $C^{(n)}_\RF$ be defined as in \eqref{capn}. Then, we have $C^{(n)}_\RF \geq \log_2(\vol(\RF))-\log_2((\pi e N_0)^n)$. \hfill $\square$
\end{proposition}
\begin{proof}
Let $p_u$ be the uniform probability distribution on $\RF$. Since $\RF$ is compact $p_u$ is well-defined. Moreover,
\begin{align*}
C^{(n)}_\RF&\geq I_{p_u}(\bfx;\bfy)= h_{p_u}(\bfy) - h(\bfw) \geq h_{p_u}(\bfy|\bfw)-h(\bfw)\nonumber\\
&=h_{p_u}(\bfx)-h(\bfw)=\log_2(\vol(\RF))- \log_2((\pi e N_0)^n).
\end{align*}
This completes the proof.
\end{proof}
A  result similar to Proposition \ref{cap_lower} holds for $\OC$. To state the lower bound in terms of the physical parameters including transmission time $T$, bandwidth $W$, amplitude $A$, and noise parameter $N_0$, we need to normalize transmission time by $T$, bandwidth by $W$, take $n=WT$, scale the harmonic basis functions $\psi_k(t)=e^{j2k \pi t}$, $t\in [0,1)$ by $\frac{1}{\sqrt{T}} \psi_k(\frac{t}{T})$, $t\in [0,T)$, normalize  the amplitude of the signal by $A\sqrt{T}$, and keep the noise parameter the same as $N_0$. Applying this normalization, and using $n!\approx \sqrt{2\pi n} (n/e)^n$, for sufficiently large $n=WT$, we obtain the following lower bounds 
for the \textit{capacity per unit time} of RF and OC channels:
\begin{align}
C_\mathsf{OC}(W)&\geq W\Big (\log_2(\frac{A^2}{2WN_0}) -4\Big)\label{cap_oc}\\
C_\mathsf{RF}(W)&\geq W\Big (\log_2(\frac{A^2}{2WN_0})+3 - 2\log_2(\alpha)\Big)\label{cap_rf}.
\end{align} 
The loss $4$ in \eqref{cap_oc} mainly results from Rogers-Shephard inequality that $\vol(\clA-\clA)\leq {4n \choose 2n} \vol(\clA)$ and can be further improved. At least, it can be reduced to $3.03=4-2\log_2(4/\pi)$ if the Mahler conjecture is true.
%

We have the additive term $3\approx \log_2(\frac{1}{e \theta^2})$ in \eqref{cap_rf} but the best value that we could find for $\alpha$ in \eqref{cap_rf} is given by $\log_2(WT)$, which results from the lower bound in Proposition \ref{logn_volume}. 
It seems that finding a universal $\alpha$ independent of $WT$ may not be possible. 
This suggests that the true behaviour of the capacity of the amplitude-limited RF channel in the \textit{one-shot} regime is of the form
$\lambda W\log_2(\frac{A^2}{2WN0})$ for some $\lambda \in (0,1)$, thus, indicating that a loss of DoFs with respect to its power-limited counterpart is unavoidable. 
This would also suggest that an effective and more practical way to approach the \textit{one-shot} capacity of this channel consists in fixing $WT$ to some sufficiently large value, and reserving a fixed fraction of subcarriers 
to keep the signal's PAPR under control. This approach, known as Tone Reservation (TR), has been widely investigated in the literature (see refs. in Section \ref{inro}) and will be treated in a novel way in Section \ref{TR}, by exploiting our 
SDP characterization of the set $\RF$.
%
%

\subsection{Upper Bound on the Capacity}
Using the results in \cite{you2002upper} and upper bounds on the volume in \eqref{vol_up_bound} derived in Section \ref{vol_section}, we obtain an upper bound on the capacity per unit time of OC in high-SNR regime, which also gives an upper bound for the capacity of RF. 
After suitable scaling we have
\begin{align}\label{cap_formula_oc}
C_\mathsf{RF}(W)\leq C_\mathsf{OC}(W) \leq W\log_2(\frac{A^2}{2 W N_0}),
\end{align} 
which shows that the lower bounds in \eqref{cap_oc} is tight for OC up to a finite loss in SNR.

\section{Tone Reservation and the PAPR Problem}\label{TR}
In this section, we investigate tone-reservation algorithm (TR) for an individual OFDM block of large dimension $n=WT$. 
Of course, in practice,
coding across a sufficiently large number of OFDM blocks can be used to further increase the reliability.  
From capacity result in \eqref{cap_rf}, it seems that the loss in SNR given by $\alpha$ might scale as $\log_2(WT)$. Thus,  to compensate the loss in SNR, it might be necessary to encode over blocks with smaller $WT$. For a fixed bandwidth $W$, it implies that the OFDM packets should be made smaller, and coding over consecutive blocks should be used to achieve the optimal scaling $W\log_2(\frac{A^2}{2 W N_0})$. 

Recently, Ilic and Strohmer \cite{ilic2009papr} studied the performance of Tone Reservation (TR) using the results of Lyubarskii and Vershynin in \cite{lyubarskii2010uncertainty}. Although not yet rigorous, their results suggest that a constant PAPR for all codewords is possible provided that a fixed fraction of carriers are devoted to waveform shaping, and PAPR reduction. This shows that, the DoF $WT$ up to a multiplicative loss seems to be achievable.
In this section, using the SDP representation for $\RF$, we prove that the best PAPR for TR is given by the statistical width of an appropriate convex set that we define. 

In TR an OFDM block of length $n$ is divided into two sub-blocks:  a block $\clI\subset [n]$ of size $m=\lambda n$, $\lambda \in (0,1)$, containing the information symbols, and a block $\clR=[n]\backslash \clI$ containing the symbols used for waveform shaping to reduce the PAPR. This reduces the rate by a factor $\lambda$. Moreover, an extra power is transmitted for the redundant symbols. Let $\bfs\in \bC^m$ be the sequence of information symbols. We assume that each component of $\bfs$ is selected from a given signal constellation such as QAM. Using the SDP representation of $\RF$ in \eqref{eq:psd_char}, and after suitable normalization, we obtain the following SDP for the optimal selection of redundant symbols:
\begin{align}\label{pts_opt_rf}
\bfx^*&=\argmin_{\bfx \in \bC^n}\ \trace[\bfH] \text{ subject to } \\
&\left [\begin{matrix} \bfH & \bfx \\ \bfx^\herm & 1 \end{matrix}\right ] \succeq 0, \sum_{\ell=1}^{n-i} \bfH_{\ell,\ell+i}= 0, i\in [n-1],\ \bfx_{\clI}=\bfs,\nonumber
\end{align}
where $\bfx_\clI$ is a sub-vector of $\bfx$ containing the components in location $\clI$. It is not difficult to check that the minimum power loss for transmitting $\bfs$ is given by $\ell(\bfs, \clI):= \frac{\trace[\bfH^*]}{\|\bfs\|^2}$, where $\bfH^*$ is the optimal matrix obtained from \eqref{pts_opt_rf}. We suppose that each symbols $s_i$ in $\bfs$ is generated i.i.d. from a given distribution $p_s$, where $\bE[s_i]=0$ and $\bE[|s_i|^2]=1$. For large $n$, we have that $\ell(\bfs,\clI)\approx \trace[\bfH^*]/m$. Note that $\ell(\bfs, \clI)$ is a random variable depending on the information symbols $\bfs$ and their location $\clI$.

Let $\clC\subset \bC^m$ be a convex set. We define the statistical width of $\clC$ under i.i.d. sampling induced by distribution $p_s(s)$ by
$
\omega(\bfs,\clC):=\sup_{\bfu\in \clC} |\inpr{\bfu}{\bfs}|^2 
$, 
where $\inpr{\bfu}{\bfs}=\Re[\bfu^\herm \bfs]$, and its average by $\overline{\omega}(\clC)=\bE_{\bfs}[\omega(\bfs, \clC)]$. When $p_s$ is a Gaussian distribution, this is known as the squared Gaussian width of the set $\clC$, which plays a crucial role in characterizing the minimum number of measurements in high-dimensional estimation under constraints such as Compressed Sensing \cite{chandrasekaran2012convex, vershynin2014estimation}. 
For a given position of information symbols $\clI$, let us define
\begin{align}\label{conv_set_rf}
\clC_\clI=\{\bfu \in \bC^m: \exists\, \bfy\ \text{ s.t. }  \left [\begin{matrix} \bT[1;\bfy] & \bfz(\bfu) \\ \bfz(\bfu)^\herm & 1 \end{matrix}\right ] \succeq 0\},
\end{align}
where $\bT[1;\bfy]$ denotes a Hermitian Toeplitz matrix with first column $[1;\bfy]$, and where $\bfz(\bfu)\in \bC^n$ is a vector containing $\bfu$ in indices corresponding to $\clI$ and zero elsewhere. It is not difficult to check that $\clC_\clI$ is a convex set containing the origin.
Then, we obtain the following proposition. 
\begin{proposition}
Let $\ell(\bfs, \clI)$ be the power-loss (PAPR) as defined before. Then, $\ell(\bfs, \clI)=\omega(\bfs,\clC_\clI)/m$, where $\omega(\bfs,\clC_\clI)$ is the width of $\clC_\clI$ for the $\bfs$ with i.i.d. components $\sim p_s$. \hfill $\square$
\end{proposition}
\begin{proof}
We use the duality to prove the result. First note that the Slater's condition \cite{boyd2004convex} holds for the SDP constraint because for sufficiently large $\tau\in \bR_+$, by setting $\bfH=\tau \bfI_n$, the matrix $\left [\begin{matrix} \bfH & \bfx \\ \bfx^\herm & 1 \end{matrix}\right ]$ will satisfy the constraints and will be strictly positive, i.e., will lie in the relative interior of the convex constrained set. Thus, we have the strong duality without any duality gap. Introducing dual variables $\bfz \in \bC^n$ for the constraint $\bfx_\clI=\bfs$, $\bfy \in \bC^{n-1}$ for the constraints in $\bfH$, and $w$ for the constraint of having $1$ at the intersection of the last row and last column, we obtain the dual of the SDP in \eqref{pts_opt_rf} as follows 
\begin{align*}
\sup_{\bfz\in \bC^{n}} -2\inpr{\bfz_\clI}{\bfs}-w \text{ s.t. } \left [\begin{matrix} \bT[1;\bfy] & \bfz \\ \bfz^\herm & w \end{matrix}\right ] \succeq 0,\  \bfz_{\clR}=0,
\end{align*}
where $\bT[1;\bfy]$ is a $n\times n$ Hermitian Toeplitz matrix whose first column is $[1;\bfy]$, where $[1;\bfy]$ is the column vector obtained by appending $1$ to $\bfy$, and where $\clR=[n]\backslash \clI$ is the position of redundant symbols.  This can be equivalently written as 
\begin{align*}
\sup_{\bfz\in \bC^{n},w} -2\sqrt{w} \inpr{\bfz_\clI}{\bfs}-w \text{ s.t. } \left [\begin{matrix} \bT[1;\bfy] & \bfz \\ \bfz^\herm & 1 \end{matrix}\right ] \succeq 0,  \bfz_{\clR}=0.
\end{align*}
We can simply check that the SDP constraints are symmetric, i.e., if $\bfz$ satisfies the constraint so does $-\bfz$. Hence, we have
\begin{align*}
\sup_{\bfz\in \bC^{n},w} 2\sqrt{w} |\inpr{\bfz_\clI}{\bfs}|-w \text{ s.t. } \left [\begin{matrix} \bT[1;\bfy] & \bfz \\ \bfz^\herm & 1 \end{matrix}\right ] \succeq 0,  \bfz_{\clR}=0.
\end{align*}
Optimizing first with respect to the variable $w$, which gives the optimal value $|\inpr{\bfz_\clI}{\bfs}|^2$, then with respect to the variable $\bfz$, and denoting $\bfz$ with the variable $\bfu \in \bC^m$ in location $\clI$, and $0$ elsewhere (i.e. $\bfz_{\clR}=0$), we obtain 
\begin{align*}
\sup_{\bfu\in \bC^{m}} |\inpr{\bfu}{\bfs}|^2 \text{ s.t. } \left [\begin{matrix} \bT[1;\bfy] & \bfz(\bfu) \\ \bfz(\bfu)^\herm & 1 \end{matrix}\right ] \succeq 0.
\end{align*}
Identifying the constrained set by $\clC_\clI$, the dual-optimal value is indeed $\omega(\bfs,\clC_\clI)$. From strong duality, this value corresponds to the primal-optimal value given by $\trace[\bfH^*]$. Thus, from definition of $\ell(\bfs, \clI)$, it is immediately seen that $\ell(\bfs, \clI)=\omega(\bfs,\clC_\clI)/m$. This completes the proof.
\end{proof}

Although further analysis is needed to study the behaviour of PAPR loss $\omega(\bfs,\clC_\clI)/m$, our numerical simulations in Section \ref{simulation} suggest that it concentrates very well around its average $\overline{\omega}(\clC_\clI)/m$, where the average seems to be a constant. 
In \cite{boche2013peak}, it was proved that if a fraction of subcarriers is reserved to reduce the PAPR so that the constant PAPR criterion is always met by \textit{all} the transmitted waveforms, then the fraction of the remaining subcarriers that can be allocated to the information symbols tends to zero. Our numerical results, however, indicate that this seems to be a worst case measure, and in fact the capacity seems to be nonzero. However, to show this rigorously, we need to prove that for any $\lambda \in (0,1)$, and for $m=n\lambda$, 
\begin{align}
\lim_{n \to \infty} \bP[\overline{\omega}(\clC_\clI)/m \leq \scrK(\lambda)]=1, 
\end{align}
for a fixed function $\scrK:(0,1) \to \bR_+$. We leave this as a future work to be investigated further.

\section{Simulation Results}\label{simulation}
Fig.~\ref{fig:papr_rf} shows the complementary density function (CCDF) of the random variable $\ell(\bfs, \clI)$ for the TR algorithm.  We assume that the transmitter uses a $16$-QAM (quadrature-amplitude modulation). We compare the results with a case in which the symbols are sampled from a circularly symmetric complex Gaussian distribution. The results show a sharp transition in CCDF for a sufficiently large redundancy (fraction of reserved tones).

\begin{figure}[!ht]
\centering
\includegraphics[width=0.5\textwidth]{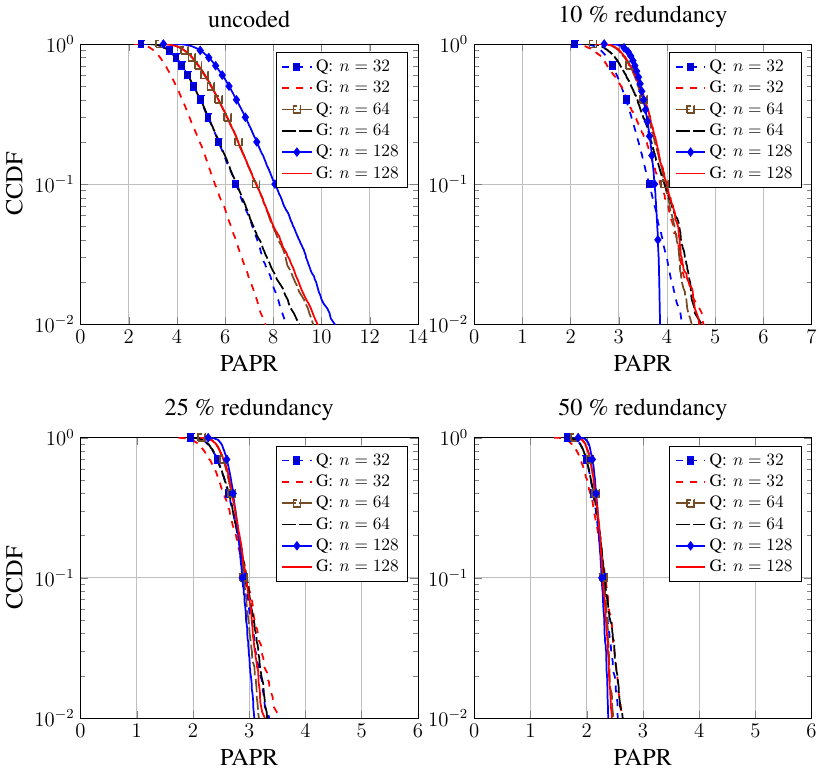}
\caption{PAPR performance of QAM (Q) and Gaussian (G) signals as a function of block-length $n$ and redundancy.}
\label{fig:papr_rf}
\end{figure}

\balance
{\footnotesize
\bibliographystyle{IEEEtran}
\bibliography{references}
}

\end{document}